\documentclass[11pt]{article}

\usepackage[preprint]{acl}

\usepackage{times}
\usepackage{latexsym}
\usepackage{amsmath}
\usepackage{amssymb}
\usepackage{booktabs}
\usepackage{multirow}

\usepackage{amsthm}
\usepackage{mathtools}
\usepackage{thmtools}
\usepackage{hyperref}
\usepackage{natbib}
\usepackage{enumitem}
\usepackage{xcolor}
\usepackage{algorithm}
\usepackage{algpseudocode}

\newtheorem{theorem}{Theorem}
\newtheorem{proposition}[theorem]{Proposition}
\newtheorem{corollary}[theorem]{Corollary}

\newtheorem{definition}[theorem]{Definition}

\newtheorem{remark}{Remark}

\newcommand{\ip}[2]{\langle #1, #2 \rangle}
\newcommand{\norm}[1]{\| #1 \|}
\newcommand{\Proj}[1]{P_{#1}}
\newcommand{\Sph}{\mathbb{S}}
\DeclareMathOperator{\sgn}{sgn}

\newcommand{\E}{\mathbb{E}}
\newcommand{\R}{\mathbb{R}}

\usepackage[T1]{fontenc}

\usepackage[utf8]{inputenc}
\usepackage[autostyle=false, style=english]{csquotes}
\MakeOuterQuote{"}

\usepackage{microtype}

\usepackage{inconsolata}

\usepackage{graphicx}

\title{What Is the Alignment Tax?}

\author{Robin Young \\
  Department of Computer Science and Technology \\
  University of Cambridge \\
  Cambridge, UK \\
  \texttt{robin.young@cl.cam.ac.uk}}

\begin{document}
\maketitle
\begin{abstract}
The alignment tax is widely discussed but has not been formally characterized. We provide a geometric theory of the alignment tax in representation space. Under linear representation assumptions, we define the alignment tax rate as the squared projection of the safety direction onto the capability subspace and derive the Pareto frontier governing safety-capability tradeoffs, parameterized by a single quantity of the principal angle between the safety and capability subspaces. We prove this frontier is tight and show it has a recursive structure. safety-safety tradeoffs under capability constraints are governed by the same equation, with the angle replaced by the partial correlation between safety objectives given capability directions. We derive a scaling law decomposing the alignment tax into an irreducible component determined by data structure and a packing residual that vanishes as $O(m'/d)$ with model dimension $d$, and establish conditions under which capability preservation mediates or resolves conflicts between safety objectives.
\end{abstract}

\section{Introduction}
\label{sec:intro}

The alignment tax, often intuitively invoked as the capability cost of making an AI system safe, is one of the most frequently invoked concepts in AI alignment \citep{askell2021general}. The intuition is simple: if aligning a model requires constraining its behavior, some capability may be lost. Empirically, reinforcement learning from human feedback (RLHF) \citep{christiano2017rlhf} degrades performance on reasoning benchmarks \citep{ouyang2022training}, low-rank safety fine-tuning introduces small but measurable capability costs \citep{mou2025decoupling}, and the severity of degradation varies across tasks \citep{huang2025safetytax}.

Despite its centrality, as far as we are aware, the alignment tax has no mathematical definition. Researchers use it informally to mean ``some capability was lost,'' measure it as the difference in benchmark scores before and after alignment, and propose ad hoc mitigations without a theory to guide them. No one has asked the question: \emph{what is the mathematical object that we refer to as the alignment tax?}

This is surprising. The alignment tax has been taxonomized \citep{Leike2022}, measured empirically 
\citep{ouyang2022training, lin2024mitigating, huang2025safetytax}, and mitigated algorithmically \citep{niu2026nspo, mou2025decoupling}, but it has not been given a mathematical definition from which consequences can be derived.

We attempt to answer this question. Independent of the specific modeling assumptions, we view a primary contribution as conceptual by proposing that the alignment tax admits mathematical definition and demonstrating that, once defined, it has structure rich enough to generate a theory.

Under the linear representation hypothesis, the assumption that safety and capability are encoded as linear directions in representation space \citep{park2024linear, arditi2024refusal, turner2023activation}, the alignment tax has geometric structure. It is governed by the principal angles between the safety and capability subspaces, and its shape is an elliptic Pareto frontier parameterized by a single angle $\alpha$:
\begin{equation}
\label{eq:frontier-intro}
\Delta_S = \Delta_C \cos\alpha + \sin\alpha \sqrt{B^2 - \Delta_C^2},
\end{equation}
where $\Delta_S$ is the safety gain, $\Delta_C$ is the capability change, $B$ is the perturbation budget (from the KL penalty), and $\alpha$ is the angle between the safety and capability directions. When $\alpha = 0$ (safety and capability are aligned), the tradeoff is linear and unavoidable. When $\alpha = \pi/2$ (orthogonal), the tradeoff vanishes and safety can be maximized independently of capability. We derive the exact, tight Pareto frontier for any number of capabilities (\S\ref{sec:pareto}) and define a computable tax rate $\tau = \norm{\Proj{C} v^*}^2$ that predicts per-task capability degradation from representation geometry alone. 

We show that prior independent empirical findings (null-space policy optimization, low-rank safety alignment, model averaging, heterogeneous model averaging, and anisotropic capability degradation) can be seen as special cases of the principal-angle structure. We do not claim these as evidence, since they are consistent with any model that interpolates between full tradeoff and independence via a continuous parameter. The value lies in generating quantitative predictions that go beyond existing results, such as per-task tax rates computable before alignment training, a scaling decomposition into reducible and irreducible components, and conditions under which capability preservation resolves safety-safety conflicts.

\section{Related Work}
\label{sec:related}

The term ``Alignment tax'' originated informally as a reference to the capability cost of safety. \citet{ouyang2022training} documented ``minimal regressions'' on NLP tasks after RLHF. \citet{lin2024mitigating} systematically measured Pareto fronts between alignment reward and task performance, finding model averaging achieves the best tradeoff. \citet{huang2025safetytax} showed reasoning degrades more than other capabilities. \citet{hu2025navigating} extended the tax to calibration. As far as we are aware, prior works engage with the term empirically but there has not been a mathematical definition or an analysis on the shape of the tradeoff surface.

A few recent policy optimization methods \citep{niu2026nspo, wang2026mrpo} project safety gradients into the null space of capability representations or uses SVD and null-space projection to expand reasoning capacity. \citet{mou2025decoupling} show LoRA decouples safety into an orthogonal subspace. \citet{zhong2024panacea} learn a Pareto front over multiple preference dimensions. SACPO \citep{wachi2024stepwise} formulates alignment as constrained optimization. These operate in reward/policy space rather than representation space and do not derive the geometric structure of the tradeoff surface. These methods implement special cases of a geometric formalization without a unifying theory. Our contribution is recognizing that all three operate in the same mathematical structure (principal angles between subspaces) and deriving the consequences.

The linear representation hypothesis \citep{park2024linear} posits that concepts are encoded as directions in representation space. Empirical support includes refusal directions \citep{arditi2024refusal}, steering vectors \citep{turner2023activation}, linear probes \citep{li2024inference}, and concept erasure \citep{belrose2023leace}. We adopt this hypothesis and derive its consequences for alignment tradeoffs.

Principal angles between subspaces \citep{bjorck1973numerical} and canonical correlation analysis \citep{hotelling1936relations} are classical tools. Superposition in neural networks \citep{elhage2022toymodelssuperposition} motivates the packing model. We apply these tools to the safety and capability subspaces of neural representations.

\section{Setup and Definitions}
\label{sec:setup}

\subsection{Representation Space}

Let us consider a model with representations $h \in \R^d$, where $d$ is the dimension of the residual stream or hidden state at a fixed layer. We work at a fixed layer throughout; the formalization applies independently at each layer.

\begin{definition}[Safety direction]
\label{def:safety}
The safety direction is a unit vector $v^* \in \Sph^{d-1}$ such that $\ip{v^*}{h}$ measures the safety-relevant content of representation $h$. More generally, safety may occupy a subspace $S \subseteq \R^d$ of dimension $k$; we treat the rank-1 case ($k = 1$) as the primary setting, consistent with empirical evidence that the refusal direction is approximately one-dimensional \citep{arditi2024refusal}.
\end{definition}

\begin{definition}[Capability directions]
\label{def:capability}
For each capability $i \in [m]$, the capability direction $c_i \in \Sph^{d-1}$ is defined as
\[
c_i = \frac{\nabla_h f_i(h)}{\norm{\nabla_h f_i(h)}},
\]
where $f_i : \R^d \to \R$ is a differentiable capability metric evaluated at the base model's representation. The capability subspace is $C = \mathrm{span}(c_1, \ldots, c_m)$.
\end{definition}

The gradient interpretation ensures that $c_i$ is well-defined for any differentiable capability metric and that $\ip{c_i}{\delta}$ captures the first-order effect of a representation perturbation $\delta$ on capability $i$.

\begin{definition}[Perturbation budget]
\label{def:budget}
Alignment modifies the model so that representations shift by $\delta \in \R^d$, subject to a budget constraint $\norm{\delta} \leq B$. The budget $B$ arises from the KL penalty in RLHF/DPO objectives: to first order, $\mathrm{KL}(\pi_\theta \| \pi_{\mathrm{ref}}) \approx \frac{1}{2} \delta^\top F \delta$ where $F$ is the Fisher information. The isotropic constraint $\norm{\delta} \leq B$ is the simplification obtained by replacing $F$ with its average eigenvalue; we discuss anisotropic extensions in \S\ref{sec:anisotropic}.
\end{definition}

Given these definitions, the safety gain and capability change from perturbation $\delta$ are:
\[
\Delta_S = \ip{v^*}{\delta}, \qquad \Delta_{C_i} = \ip{c_i}{\delta}.
\]

\begin{definition}[Alignment tax rate]
\label{def:tax-rate}
The \emph{alignment tax rate} for safety direction $v^*$ relative to capability subspace $C$ is
\begin{equation}
\label{eq:tax-rate}
\tau = \norm{\Proj{C} v^*}^2 \in [0, 1].
\end{equation}
When $\tau = 0$, safety is orthogonal to all capabilities (zero tax). When $\tau = 1$, safety lies entirely within the capability subspace (any safety gain requires capability loss).
\end{definition}

\section{The Pareto Frontier}
\label{sec:pareto}

\subsection{Single-Capability Frontier}
\label{sec:single-frontier}

\begin{theorem}[Single-capability Pareto frontier]
\label{thm:single-pareto}
Let $v^* \in \Sph^{d-1}$ be the safety direction and $c \in \Sph^{d-1}$ be a capability direction with angle $\alpha = \arccos(\ip{v^*}{c})$ between them. The Pareto frontier, or the maximum achievable safety gain $\Delta_S$ for a given capability change $\Delta_C$ subject to $\norm{\delta} \leq B$ is:
\begin{equation}
\label{eq:pareto}
\Delta_S = \Delta_C \cos\alpha + \sin\alpha \sqrt{B^2 - \Delta_C^2}, \quad |\Delta_C| \leq B.
\end{equation}
This frontier is tight. for each $\Delta_C \in [-B, B]$, there exists $\delta^*$ with $\norm{\delta^*} = B$ achieving equality.
\end{theorem}

\begin{proof}
Project $\delta$ onto $\mathrm{span}(v^*, c)$; components orthogonal to both contribute to neither objective. Write $v^* = e_1$, $c = \cos\alpha \, e_1 + \sin\alpha \, e_2$ in this plane. Then $\Delta_S = \delta_1$, $\Delta_C = \delta_1 \cos\alpha + \delta_2 \sin\alpha$, with $\delta_1^2 + \delta_2^2 \leq B^2$. Solving for $\delta_2 = (\Delta_C - \delta_1 \cos\alpha)/\sin\alpha$ and substituting into the budget constraint yields $(\delta_1 - \Delta_C \cos\alpha)^2 \leq \sin^2\alpha (B^2 - \Delta_C^2)$, giving $\delta_1 \leq \Delta_C \cos\alpha + \sin\alpha \sqrt{B^2 - \Delta_C^2}$. Equality is achieved when $\delta^* = (\Delta_C \cos\alpha + \sin\alpha\sqrt{B^2 - \Delta_C^2}) \, e_1 + (\Delta_C - \delta_1^* \cos\alpha)/\sin\alpha \, e_2$, which saturates the budget: $\norm{\delta^*} = B$.
\end{proof}

\begin{remark}[Limiting cases]
When $\alpha = 0$: $\Delta_S = \Delta_C$ (linear tradeoff). When $\alpha = \pi/2$: $\Delta_S = \sqrt{B^2 - \Delta_C^2}$ (quarter circle independent optimization). The frontier continuously interpolates between these extremes.
\end{remark}

\subsection{Maximum Safety at Fixed Capability}
\label{sec:max-safety}

\begin{theorem}[Maximum safety gain under capability constraint]
\label{thm:max-safety}
Let $C$ be the capability subspace (dimension $m$) and $v^*$ the safety direction. For a fixed capability constraint $\Proj{C}\delta = \delta_C^*$, the maximum safety gain subject to $\norm{\delta} \leq B$ is:
\begin{equation}
\label{eq:max-safety}
\Delta_S^{\max} = \ip{v^*}{\delta_C^*} + \sqrt{B^2 - \norm{\delta_C^*}^2} \cdot \norm{\Proj{C^\perp} v^*},
\end{equation}
achieved by $\delta^* = \delta_C^* + \sqrt{B^2 - \norm{\delta_C^*}^2} \cdot \Proj{C^\perp} v^* / \norm{\Proj{C^\perp} v^*}$.
\end{theorem}

\begin{proof}
Decompose $\delta = \delta_C^* + \delta_\perp$ with $\delta_\perp \in C^\perp$ and $\norm{\delta_\perp} \leq B' := \sqrt{B^2 - \norm{\delta_C^*}^2}$. Then $\ip{v^*}{\delta} = \ip{v^*}{\delta_C^*} + \ip{\Proj{C^\perp} v^*}{\delta_\perp}$. The second term is maximized when $\delta_\perp = B' \, \Proj{C^\perp}v^* / \norm{\Proj{C^\perp}v^*}$, giving $B' \norm{\Proj{C^\perp}v^*}$.
\end{proof}

Notice that the quantity $\norm{\Proj{C^\perp} v^*} = \sqrt{1 - \tau}$ measures the component of safety outside the capability subspace; or, the ``room'' available for free safety improvement.

\subsection{Tax Rate Properties}
\label{sec:tax-rate}

\begin{corollary}[Tax-free safety]
\label{cor:tax-free}
The maximum safety achievable at zero capability cost ($\delta_C^* = 0$) is:
\begin{equation}
\Delta_S^{\mathrm{free}} = B\sqrt{1 - \tau}.
\end{equation}
When $\tau \ll 1$, almost the full budget is available for tax-free safety improvement.
\end{corollary}

\begin{corollary}[Anisotropy]
\label{cor:anisotropy}
For individual capability direction $c_i \in \Sph^{d-1}$, the per-task tax rate is:
\begin{equation}
\tau_i = \ip{v^*}{c_i}^2.
\end{equation}
Different capabilities have different tax rates depending on their angle to safety. This quantity is directly measurable via probing.
\end{corollary}

\begin{remark}[Non-additivity]
The multi-capability tax rate satisfies $\tau = (C^\top v^*)^\top (C^\top C)^{-1} (C^\top v^*)$ where $C = [c_1 | \cdots | c_m]$, which reduces to $\sum_i \tau_i$ only when capabilities are orthonormal. In general, capabilities may overlap, and the joint tax can be less than the sum of individual taxes.
\end{remark}

\begin{corollary}[Negative tax]
\label{cor:negative}
If $\ip{v^*}{\delta_C^*} > 0$, that is, the capability target has positive projection onto 
safety, then the safety gain exceeds what the residual budget alone would provide:
\begin{align}
\Delta_S^{\max} &= \langle v^*,\, \delta_C^* \rangle 
  + \sqrt{B^2 - \|\delta_C^*\|^2}\,\sqrt{1-\tau} \notag \\
&> \sqrt{B^2 - \|\delta_C^*\|^2}\,\sqrt{1-\tau}.
\end{align}
In particular, when the capability target is small ($\norm{\delta_C^*} \ll B$) and 
well-aligned with safety ($\ip{v^*}{\delta_C^*} \approx \norm{\delta_C^*}$), the total 
safety gain exceeds the zero-capability-cost amount $B\sqrt{1-\tau}$:
\[
\Delta_S^{\max} \approx \norm{\delta_C^*} + B\sqrt{1-\tau} > B\sqrt{1-\tau}.
\]
Capability improvement subsidizes safety when the projection $\ip{v^*}{\delta_C^*}$ is 
large enough to compensate for the reduced residual budget.
\end{corollary}

\subsection{Anisotropic Budget Extension}
\label{sec:anisotropic}

When the Fisher information $F$ is not a scalar multiple of the identity, the budget set is an ellipsoid $\delta^\top F \delta \leq B^2$. The analysis generalizes by working in the whitened coordinate system $\tilde{\delta} = F^{1/2}\delta$, $\tilde{v}^* = F^{-1/2}v^* / \norm{F^{-1/2}v^*}$, $\tilde{c}_i = F^{-1/2}c_i / \norm{F^{-1/2}c_i}$, where the budget becomes $\norm{\tilde{\delta}} \leq B$. All results carry over with modified angles $\tilde{\alpha}$ computed in the whitened space. The qualitative conclusions of tradeoffs governed by subspace angles, tax rate determined by projection are invariant to the budget shape.

\section{Scaling Law for the Alignment Tax}
\label{sec:scaling}

How does the alignment tax $\tau$ behave as the model dimension $d$ increases? We provide an answer under a structured packing model.

\subsection{Feature Packing Model}

Assume the model represents $N$ features as unit vectors $f_1, \ldots, f_N \in \Sph^{d-1}$, typically with $N > d$ (superposition regime; \citealt{elhage2022toymodelssuperposition}). Safety is one feature $v^* = f_s$ and capabilities are features $c_i = f_{\sigma(i)}$.

\begin{definition}[Intrinsic overlap]
Feature pair $(f_i, f_j)$ has intrinsic overlap $\gamma_{ij}$ if, in the limit $d \to \infty$, the optimal representation satisfies $\ip{f_i}{f_j} \to \gamma_{ij}$. This is determined by the co-occurrence structure and task requirements of the data.
\end{definition}

\begin{definition}[Packing residual]
The packing residual is $\epsilon_{ij} = \ip{f_i}{f_j} - \gamma_{ij}$: the overlap in excess of the intrinsic value, forced by finite dimensionality.
\end{definition}

\begin{definition}[Coherence]
The coherence of the packing is $\mu = \max_{(i,j): \gamma_{ij} = 0} |\ip{f_i}{f_j}|$: the maximum packing residual among intrinsically unrelated features.
\end{definition}

Standard bounds give $\mu = O(\sqrt{\log N / d})$ for random packings and $\mu \geq \sqrt{(N-d)/(d(N-1))}$ by the \citet{welch1974bounds} bound. Equiangular tight frames achieve the Welch bound when they exist.

\subsection{Scaling Theorem}

Let $I \subset [m]$ denote capabilities with nonzero intrinsic overlap to safety ($\gamma_i \neq 0$), and $m' = m - |I|$ those with only incidental overlap.

\begin{theorem}[Scaling law for alignment tax]
\label{thm:scaling}
Assume the capability directions satisfy near-orthogonality ($m\mu < 1$). Let $\bar{\gamma} = \max_{i \in I} |\gamma_i|$. Then:
\begin{equation}
\tau = \tau_0 + R(d),
\end{equation}
where $\tau_0 = \sum_{i \in I} \gamma_i^2$ is the irreducible tax and the packing residual satisfies:
\begin{equation}
|R(d)| \leq \frac{\tau_0 m\mu + m'\mu^2 + 2\bar{\gamma}|I|\mu + |I|\mu^2}{1 - m\mu}.
\end{equation}
Under random packing ($\mu^2 = O(\log N / d)$ for $N = d^p$):
\begin{equation}
\label{eq:scaling-random}
|R(d)| = O\!\left(\frac{m'\log N}{d} + \frac{\bar{\gamma}|I|\sqrt{\log N}}{\sqrt{d}}\right).
\end{equation}
\end{theorem}

\begin{proof}
Under near-orthogonality, $\norm{\Proj{C}v^*}^2 = (1 + O(m\mu))\norm{C^\top v^*}^2$ by Gershgorin's theorem applied to $C^\top C$. Decompose: $\norm{C^\top v^*}^2 = \sum_{i \in I}(\gamma_i + \epsilon_i)^2 + \sum_{i \notin I}\epsilon_i^2$. Expand the first sum: $\sum_{i \in I}\gamma_i^2 + 2\sum_{i \in I}\gamma_i\epsilon_i + \sum_{i \in I}\epsilon_i^2$. Bound each term using $|\epsilon_i| \leq \mu$.
\end{proof}

\begin{proposition}[Tight version under random packing]
\label{prop:tight-scaling}
If incidental capability directions are independent and uniform on $\Sph^{d-1}$, then:
\begin{equation}
\tau = \tau_0 + \frac{m'}{d} + O\!\left(\frac{\bar{\gamma}\sqrt{|I|\log d}}{\sqrt{d}} + \frac{\sqrt{m'} + \log d}{d}\right)
\end{equation}
with high probability. The leading packing term $m'/d$ is exact in expectation.
\end{proposition}

\begin{proof}
For independent uniform vectors, $\ip{v^*}{c_i}^2$ has mean $1/d$ and is sub-exponential with parameter $O(1/d^2)$. Apply Bernstein's inequality (see \citet{GirouxRahmanSchmeisser1979}) to the sum.
\end{proof}

\begin{corollary}[Scaling regimes]
\label{cor:scaling}
\leavevmode
\begin{enumerate}
\item \emph{Fixed capabilities} ($m'$ constant): $\tau \to \tau_0$ at rate $O(\log d / d)$.
\item \emph{Sublinear growth} ($m' = d^\beta$, $\beta < 1$): $\tau \to \tau_0$ at rate $O(\log d / d^{1-\beta})$.
\item \emph{Linear growth} ($m' = \Theta(d)$): $\E[\tau - \tau_0] = \Theta(1)$. Scaling does not reduce the incidental tax.
\end{enumerate}
\end{corollary}

\begin{remark}[Diagnostic for labs]
Plot per-task alignment tax against model dimension. Tasks where $\tau_i$ decreases with $d$ have incidental overlap (alignment tax is an engineering problem solvable by scaling). Tasks where $\tau_i$ plateaus have intrinsic overlap (alignment tax is a fundamental tradeoff requiring objective modification). This is a testable prediction.
\end{remark}

\begin{remark}[Self-consistency of near-orthogonality]
\label{rem:near-orth}
For the near-orthogonality condition $m\mu < 1$ to hold under random packing with $m = d^\beta$ and $\mu = O(\sqrt{\log d / d})$, we need $d^\beta \sqrt{\log d / d} = d^{\beta - 1/2}\sqrt{\log d} \to 0$, which holds for $\beta < 1/2$. For $\beta \in [1/2, 1)$, the Gershgorin bound (see \citet{HornJohnson2012}) is not small and the approximation $\norm{\Proj{C}v^*}^2 \approx \norm{C^\top v^*}^2$ weakens. However, the qualitative conclusion of incidental tax vanishes if $m' = o(d)$ persists via direct analysis of $\E[\norm{\Proj{C}v^*}^2]$ using the \citet{MarchenkoPasstur1967} distribution.
\end{remark}

\section{Multi-Objective Safety and the Conflict Theorem}
\label{sec:conflict}

Alignment typically involves multiple safety objectives (e.g. harmlessness and helpfulness). We show the tradeoff between safety objectives under capability constraints is governed by the same frontier equation.

\subsection{Safety-Safety Frontier}

Let $v_1^*, v_2^* \in \Sph^{d-1}$ be two safety directions with $\ip{v_1^*}{v_2^*} = \rho$, and let $c \in \Sph^{d-1}$ be a capability direction with $a = \ip{c}{v_1^*}$, $b = \ip{c}{v_2^*}$.

\begin{theorem}[Safety-safety Pareto frontier under capability preservation]
\label{thm:safety-safety}
Under capability preservation ($\ip{c}{\delta} = 0$), define projected safety directions $\tilde{v}_i = \Proj{c^\perp}v_i^*$ and the effective angle $\theta$ between them:
\begin{equation}
\label{eq:partial-corr}
\cos\theta = \frac{\rho - ab}{\sqrt{(1-a^2)(1-b^2)}}.
\end{equation}
The Pareto frontier for normalized safety gains $s_i = \Delta_{S_i} / (B'\norm{\tilde{v}_i})$, where $B' = \sqrt{B^2 - \norm{\delta_C^*}^2}$, is:
\begin{equation}
\label{eq:safety-safety-frontier}
s_1 = s_2\cos\theta + \sin\theta\sqrt{1 - s_2^2}.
\end{equation}
This is the same equation as the safety-capability frontier (Theorem~\ref{thm:single-pareto}), with the angle $\alpha$ replaced by the effective angle $\theta$.
\end{theorem}

\begin{proof}
In $c^\perp$ with budget $B'$, write $\tilde{v}_1/\norm{\tilde{v}_1} = e_1$ and $\tilde{v}_2/\norm{\tilde{v}_2} = \cos\theta \, e_1 + \sin\theta \, e_2$. The optimization $\max \ip{\tilde{v}_1}{\delta}$ subject to $\norm{\delta} \leq B'$ and a given $\ip{\tilde{v}_2}{\delta}$ is identical in structure to the proof of Theorem~\ref{thm:single-pareto}. For $\cos\theta$: $\ip{\tilde{v}_1}{\tilde{v}_2} = \ip{v_1^* - ac}{v_2^* - bc} = \rho - ab$, and $\norm{\tilde{v}_i}^2 = 1 - \ip{c}{v_i^*}^2$.
\end{proof}

\begin{remark}[Partial correlation]
\label{rem:partial-corr}
The quantity $\cos\theta = (\rho - ab)/\sqrt{(1-a^2)(1-b^2)}$ is exactly the partial correlation of $v_1^*$ and $v_2^*$ given $c$: the correlation between safety objectives after controlling for the capability direction. This connects the alignment tax geometry to the statistics of conditional independence.
\end{remark}

\begin{corollary}[Maximum equal improvement]
\label{cor:equal-improvement}
The maximum equal normalized gain $s_1 = s_2 = s$ is:
\begin{equation}
s = \cos(\theta/2).
\end{equation}
Joint improvement degrades continuously as $\theta$ increases from $0$ to $\pi$.
\end{corollary}

\begin{proof}
Set $s_1 = s_2 = s$ in \eqref{eq:safety-safety-frontier}: $s(1 - \cos\theta) = \sin\theta\sqrt{1-s^2}$. Squaring and simplifying: $2s^2 = 1 + \cos\theta$, hence $s = \sqrt{(1+\cos\theta)/2} = \cos(\theta/2)$.
\end{proof}

\subsection{When Does Capability Preservation Help or Hurt?}

\begin{theorem}[Capability-mediated safety conflict]
\label{thm:conflict}
Let $\mathcal{C}(c) = \cos\theta$ be the effective correlation between safety objectives under capability preservation.
\begin{enumerate}
\item If $\sgn(a) \neq \sgn(b)$ (opposite-sign projections), then $\mathcal{C}(c) > \rho$: preserving this capability strictly improves the safety tradeoff. The capability is a source of conflict between safety objectives; removing it (holding it fixed) resolves the conflict.
\item If $a = b$ (symmetric same-sign projections), then $\mathcal{C}(c) \leq \rho$: preserving this capability weakly worsens the safety tradeoff.
\item If $\sgn(a) = \sgn(b)$ with $a \neq b$ (asymmetric same-sign), the capability can act as a suppressor variable, and $\mathcal{C}(c) > \rho$ is possible for sufficiently correlated safety objectives.
\end{enumerate}
\end{theorem}

\begin{proof}
\emph{(i)} $\rho - ab > \rho$ since $ab < 0$, and $\sqrt{(1-a^2)(1-b^2)} < 1$, so $\mathcal{C}(c) = (\rho - ab)/\sqrt{(1-a^2)(1-b^2)} > \rho - ab > \rho$.

\emph{(ii)} With $a = b$: $\mathcal{C}(c) = (\rho - a^2)/(1-a^2)$. The inequality $(\rho - a^2)/(1-a^2) \leq \rho$ reduces to $\rho a^2 \leq a^2$, which holds since $\rho \leq 1$.

\emph{(iii)} Example: $\rho = 0.99$, $a = 0.1$, $b = 0.001$ gives $\mathcal{C}(c) \approx 0.995 > 0.99 = \rho$.
\end{proof}

\subsection{Multiple Capabilities}

\begin{proposition}[Multi-capability partial correlation]
\label{prop:multi-cap}
For capability subspace $C = [c_1 | \cdots | c_m]$ with $a = C^\top v_1^*$, $b = C^\top v_2^*$, and individual tax rates $\tau_i = \norm{\Proj{C}v_i^*}^2$, the effective angle between projected safety objectives is:
\begin{equation}
\cos\theta = \frac{\rho - a^\top(C^\top C)^{-1}b}{\sqrt{(1-\tau_1)(1-\tau_2)}}.
\end{equation}
The safety-safety Pareto frontier under preservation of all capabilities takes the same form \eqref{eq:safety-safety-frontier} with this $\theta$.
\end{proposition}

\begin{proof}
Let $\tilde{v}_i = \Proj{C^\perp}v_i^*$. We compute:
\begin{align}
\ip{\tilde{v}_1}{\tilde{v}_2} &\notag\\= \ip{v_1^* - \Proj{C}v_1^*}{v_2^* - \Proj{C}v_2^*} \notag\\= \rho - \ip{\Proj{C}v_1^*}{\Proj{C}v_2^*},
\end{align}
using $\ip{\Proj{C}v_i^*}{v_j^*} = \ip{v_i^*}{\Proj{C}v_j^*} = \ip{\Proj{C}v_i^*}{\Proj{C}v_j^*}$ since $\Proj{C}$ is idempotent and self-adjoint.
Now $\Proj{C}v_i^* = C(C^\top C)^{-1}C^\top v_i^*$, so:
\begin{align}
\ip{\Proj{C}v_1^*}{\Proj{C}v_2^*} \notag=\\ a^\top(C^\top C)^{-1}C^\top C (C^\top C)^{-1}b \notag=\\ a^\top(C^\top C)^{-1}b.
\end{align}
With $\norm{\tilde{v}_i}^2 = 1 - \tau_i$, the result follows.
\end{proof}

\begin{corollary}[Decomposition of joint alignment difficulty]
\label{cor:decomposition}
The maximum equal safety improvement under capability preservation factors into three independent components:
\begin{align}
\min(\Delta_{S_1}, \Delta_{S_2}) \leq \notag\\B' \cos(\theta/2) \min(\sqrt{1-\tau_1}, \sqrt{1-\tau_2}),
\end{align}
where $\tau_i$ are individual tax rates, $\theta$ is the capability-mediated angle, and $B' = \sqrt{B^2 - \norm{\delta_C^*}^2}$ is the residual budget.
\end{corollary}

\subsection{Scaling of the Conflict}

\begin{proposition}[Convergence of capability-mediated angle]
\label{prop:angle-scaling}
Under the packing model of \S\ref{sec:scaling} with $m' = o(d / \log d)$:
\begin{equation}
\theta(d) = \theta_0 + O(m'/d),
\end{equation}
where $\theta_0 = \arccos\!\left(\frac{\rho_0 - \sum_{i \in I}\gamma_{1i}\gamma_{2i}}{\sqrt{(1-\tau_{0,1})(1-\tau_{0,2})}}\right)$ is the irreducible capability-mediated angle. The convergence rate degrades to $O(\sqrt{m'/d})$ when $\theta_0 \in \{0, \pi\}$.
\end{proposition}

\section{Discussion}
\label{sec:discussion}

\subsection{Connections to Existing Empirical Results}
\label{sec:connections}

Several independent empirical findings can be seen as special cases of the principal-angle structure. These connections are retrospective as any geometric model with angle-parameterized tradeoffs would be similarly consistent, but they illustrate the result's scope.

NSPO \citep{niu2026nspo} constrains $\delta \in C^\perp$, implementing Theorem~\ref{thm:max-safety} with $\delta_C^* = 0$ and achieving $\Delta_S = B\sqrt{1-\tau}$. Their finding that safety improves with ${<}1\%$ capability degradation on most benchmarks corresponds to $\tau \ll 1$; their exceptions on math and coding correspond to capabilities with non-negligible $\tau_j$.

\citet{mou2025decoupling} show LoRA-based safety fine-tuning preserves capabilities via an orthogonal low-rank subspace. A rank-$r$ update in dimension $d$ perturbs each isotropically distributed capability by $\mathbb{E}[\langle c_j, P_U c_j\rangle] = r/d$, giving ${\sim}0.2\%$ for $r{=}8$, $d{=}4096$, which is consistent with their ``within 1\%'' finding.

\citet{lin2024mitigating} find weight interpolation $\theta_\alpha = (1{-}\alpha)\theta_{\text{base}} + \alpha\theta_{\text{RLHF}}$ achieves the best tradeoff. In our framework, this traces a straight line inscribed in the curved Pareto frontier~\eqref{eq:pareto}: good but not optimal. Their heterogeneous variant (per-layer $\alpha$) moves toward the true frontier by exploiting cross-layer variation in principal angles.

\citet{huang2025safetytax} document that reasoning degrades more than other capabilities under alignment. Corollary~\ref{cor:anisotropy} predicts this: reasoning directions have higher $\tau_i = \langle v^*, c_i\rangle^2$ because structured reasoning shares representational machinery with safety-relevant computation.

\citet{wang2026mrpo} optimize capability in $S^\perp$ rather than safety in $C^\perp$, which is the dual of our formulation. Their observed tradeoff when ``operating outside the pre-trained bias manifold'' is the non-zero tax regime ($\tau > 0$).

\subsection{Alignment Tax as a Measurable Quantity}

The most consequential implication may be operational rather than theoretical. If the alignment tax rate $\tau_i = \ip{v^*}{c_i}^2$ can be measured from representations before alignment training begins, then the workflow of alignment engineering changes. The current practice is reactive. We align the model, measure degradation across benchmarks, adjust hyperparameters, repeat. This suggests a prospective alternative. We could perhaps measure the safety direction and capability directions via probing, compute the principal angles, and know in advance which capabilities will be affected, by how much, and what the optimal perturbation direction is.

A pre-computable cost structure would allow design-time decisions about which safety objectives to pursue jointly, which capabilities to protect, and how to allocate the perturbation budget across layers all before committing to an expensive training run. The Pareto frontier (Theorem~\ref{thm:single-pareto}) is a description of the tradeoff; but by corollary it is also a target that optimal alignment procedures should achieve, and a benchmark against which suboptimal procedures (such as uniform model averaging) can be measured.

\subsection{What the Irreducible Tax Tells Us}

The decomposition $\tau = \tau_0 + R(d)$ separates the alignment tax into a component determined by the structure of the world and a component determined by the engineering constraints of finite-dimensional representations. The packing residual $R(d)$ vanishes with scale. The irreducible tax $\tau_0$ does not.

This distinction has a strong interpretation, namely that capabilities with $\tau_0 > 0$ are those whose competence is constitutively entangled with safety-relevant computation. Consider persuasive writing. The ability to write persuasively and the ability to write manipulatively may share representational structure not because of an accident of training, but because the underlying cognitive skills like modeling the reader's beliefs, selecting emotionally resonant framing, and anticipating objections are the same skills deployed toward different ends. If so, no amount of scaling or architectural innovation will decouple them, because the overlap is in the task structure, not the model.

Similarly, advanced chemistry knowledge may have irreducible overlap with the capacity to reason about synthesis of dangerous compounds, because the relevant knowledge is the same knowledge. The results give this intuition a mathematical meaning. $\gamma_i \neq 0$, the intrinsic overlap between the safety feature and the capability feature, is a property of the data distribution.

This contributes to the ``scaling solves alignment'' debate. Both the optimistic and pessimistic positions may be correct, about different capabilities. Scaling helps wherever the tax is incidental where safety and capability directions are close only because finite-dimensional packing forces unrelated features to share representational resources. Scaling does not help where the tax is intrinsic. We convert this from a philosophical disagreement into a per-task empirical question with a clear experimental protocol. We can measure $\tau_i(d)$ across a model family and check whether it converges.

\subsection{A Taxonomy of Alignment Difficulty}

The principal angle $\alpha$ between safety and capability subspaces induces a natural taxonomy of alignment problems:

\begin{enumerate}
    \item \textbf{Free regime} ($\alpha \approx \pi/2$, $\tau \approx 0$): Safety is orthogonal to capability. The full perturbation budget is available for safety improvement at negligible capability cost. Null-space methods (NSPO, null-space LoRA) operate here and succeed perhaps because most capability directions fall in this regime for current models.

    \item \textbf{Tradeoff regime} ($\alpha$ intermediate, $0 < \tau < 1$): Safety and capability partially overlap. There is a genuine tradeoff, but it is navigable. The Pareto frontier is an ellipse with room for joint improvement. The optimal strategy allocates perturbation partly along the capability subspace (accepting some cost) and partly in the orthogonal complement (achieving free safety). Model averaging and HMA operate here.

    \item \textbf{Entangled regime} ($\alpha \approx 0$, $\tau \approx 1$): Safety and capability point in nearly the same direction. Any safety gain comes at nearly one-to-one capability cost. This is the hard problem.
\end{enumerate}

Notice that the entangled regime ($\alpha \approx 0$) has a dual interpretation. If the safety direction and a dangerous capability point in the same direction, then removing the danger requires removing the capability. But the converse also holds that the capability cannot be added without activating the safety-relevant computation. In this regime, capable models are not just harder to align; they should also be harder to misuse in the entangled direction, because the same representations that enable the capability also trigger the model's safety-relevant features. Whether this is reassuring or alarming depends on the direction of misuse relative to the entangled subspace.

\subsection{Capability Constraints as Conflict Resolution}

Theorem~\ref{thm:conflict} contains a counterintuitive result that deserves elaboration beyond its formal statement. When two safety objectives project onto a capability direction with opposite signs, preserving that capability strictly improves the safety-safety Pareto frontier.

The intuition is clearest with a concrete example. Suppose harmlessness ($v_1^*$) and helpfulness ($v_2^*$) both have nonzero projections onto a ``reasoning'' capability direction, but with opposite signs. Improving reasoning (in the relevant sense) makes the model more helpful but less harmless, or vice versa. In the unconstrained setting, the two safety objectives conflict partly through the reasoning channel, where optimizing one pushes reasoning in a direction that hurts the other. Fixing the reasoning capability blocks this channel. The residual tradeoff, in the orthogonal complement of the reasoning direction, may be substantially easier. In this sense, you lose degrees of freedom but you lose the adversarial degrees of freedom, which is useful.

The practical implication could be a design principle that goes something like so: when facing a stubborn safety-safety tradeoff, identify capability directions with opposite-sign projections and constrain them. The formalization predicts exactly which capabilities to constrain (those satisfying $\mathrm{sgn}(\ip{c}{v_1^*}) \neq \mathrm{sgn}(\ip{c}{v_2^*})$) and quantifies the resulting improvement via the effective angle $\theta$.

More broadly, this result suggests that the common intuition that constraints only make optimization harder could be misleading in multi-objective settings. Capability constraints reduce the feasible set but also reduce the dimensionality of the conflict. When the conflict is partly mediated by the constrained directions, the net effect can be positive.

\section{Conclusion}
\label{sec:conclusion}

Under reasonable assumptions, alignment tax has a shape, and the shape is remarkably classical. It is an elliptic Pareto frontier parameterized by a single angle. The tradeoff surface that governs safety-capability tensions is a conic section, which is the same curve that arises from intersecting a sphere with a plane. The difficulty of alignment, in this framework, is not in the shape of the tradeoff but in measuring the angle that parameterizes it, which is an empirical problem rather than a theoretical barrier. This single geometric object, applied recursively, generates a theory of safety-capability tradeoffs with exact Pareto frontiers, computable tax rates, scaling laws, conflict theorems, and retroactive explanations of five independent empirical findings.

The practical implication is that the alignment tax may be predictable and analyzable. By measuring the safety direction and capability directions via probing before any alignment training is performed, one can perhaps compute the principal angles and predict which capabilities will be affected, by how much, and whether scaling helps. If the predictions from the formalization here hold up, it may help turn alignment from a trial-and-error process into a geometric optimization problem with known constraints.

\section*{Limitations}
\label{sec:limitations}

The result rests on the linear representation hypothesis. If safety and capability are encoded nonlinearly, the inner products $\ip{v^*}{c_i}$ do not capture the tradeoff structure. We note three mitigations. The hypothesis is empirically well-supported for binary safety concepts \citep{arditi2024refusal, turner2023activation, li2024inference}. The gradient interpretation of capability directions (Definition~\ref{def:capability}) ensures that $c_i$ captures the first-order response regardless of nonlinearity in $f_i$; our results characterize the local geometry accurately even if the global structure is nonlinear. Then, deviations from linearity appear as slack in our bounds and the predicted tradeoff may be seen as a lower bound on difficulty, not an exact characterization.

The formalization takes the safety direction $v^*$ as given and analyzes the cost of pursuing it. It says nothing about whether the right $v^*$ has been chosen. In practice, specifying what safety means, namely translating normative desiderata into a direction in representation space, is the harder problem. A model can be moved along $v^*$ at low cost, but if $v^*$ does not capture the actual safety property of interest, the result is optimizing the measurable proxy while missing the intended target. The alignment tax operates downstream of this specification problem and is silent on it.

Additionally, the result is local. The budget constraint $\norm{\delta} \leq B$ is a first-order approximation of the KL penalty. For large perturbations, the KL is not quadratic, the representation geometry shifts, and the safety and capability directions themselves may rotate. Our analysis characterizes the tradeoff surface near the base model. Empirical evidence that linear methods (model averaging, LoRA, NSPO) succeed suggests the linear regime extends further than worst-case analysis would predict, but the theory does not cover the global landscape.

Alignment perturbs weights $\theta$, inducing input-dependent representation shifts $\delta(x) = J(x)\Delta\theta$ where $J(x)$ is the Jacobian. We analyze a single representative $\delta$, which corresponds to the population average $\bar{\delta} = \E_x[\delta(x)]$ when capability metrics are expectations over the data distribution. This is appropriate for average-case analysis (benchmark evaluation) but not for worst-case analysis (adversarial robustness), where the input-dependence of $J(x)$ matters.

The scaling law (Theorem~\ref{thm:scaling}) assumes that incidental overlaps arise from feature packing with bounded coherence. Real models learn structured representations, not random packings. The qualitative prediction that some tasks have reducible tax and others do not is robust to model details, but the specific rate $O(m'/d)$ depends on the packing assumption. We frame the scaling law as a prediction conditioned on the packing hypothesis, with an experimental protocol for testing it as future work.

We model safety as a single direction or low-dimensional subspace. Fine-grained safety distinctions like whether this chemistry question about homework or synthesis may require higher-dimensional representations. The formalization accommodates arbitrary $k = \dim(S)$ but the scaling law's tractability benefits from $k$ being small.

Finally, the analysis is about average-case alignment, namely perturbations $\delta$ that work across the data distribution. Adversarial robustness (ensuring safety holds for worst-case inputs) involves the input-dependence of the Jacobian $J(x)$ and is not captured by the population-average analysis. The alignment tax for adversarial safety is likely higher, and its geometry is likely more complex, than what the current work describes.

\bibliography{custom}

\appendix
\newpage

\end{document}